\documentclass[10pt,twocolumn]{article}

\usepackage[utf8]{inputenc}
\usepackage[T1]{fontenc}
\usepackage[expansion=false]{microtype}
\usepackage[margin=0.85in]{geometry}
\usepackage{amsmath,amssymb,amsthm}
\theoremstyle{plain}
\newtheorem{lemma}{Lemma}
\newtheorem{proposition}{Proposition}
\newtheorem{conjecture}{Conjecture}
\theoremstyle{remark}
\newtheorem{remark}{Remark}
\usepackage{graphicx}
\usepackage{booktabs}
\usepackage{caption}
\usepackage{subcaption}
\usepackage{placeins}
\usepackage{xcolor}
\usepackage{algorithm}
\usepackage{algpseudocode}
\usepackage[colorlinks=true,linkcolor=blue,citecolor=blue,urlcolor=blue]{hyperref}

\graphicspath{{figures/}{../results/}}

\newcommand{\KL}{\mathrm{KL}}

\newcommand{\xbar}{\bar{x}}
\newcommand{\ybar}{\bar{y}}
\newcommand{\softmax}{\mathrm{softmax}}

\title{\textbf{GARIP: A Running-Average Moving Reference for\\
Last-Iterate Self-Play in Two-Player Zero-Sum Games}}

\author{Can Savcı\\
Independent Researcher\\
İzmir, Türkiye\\
\texttt{savcisavcican@gmail.com}}
\date{}

\begin{document}
\maketitle

\begin{abstract}
Self-play with naive gradient ascent cycles in two-player zero-sum games: the last
iterate orbits the equilibrium. Modern methods restore last-iterate convergence by
regularizing toward a \emph{reference} policy --- MMD a \emph{fixed} one (reaching
only the regularized equilibrium), R-NaD a \emph{periodic snapshot} (the engine of
DeepNash). We study \textbf{GARIP}, which anchors to the \emph{running average}, and
isolate what the choice of reference controls. Our central result is a mechanism:
collapse tracks the \emph{peak} lag of the reference, and among causal convex averages
of a fixed mean lag the running average (flat profile, peak $=$ mean) uniquely
\emph{minimizes} that peak, while a snapshot's sawtooth has peak $=2\times$ mean (a
one-line theorem). Two consequences follow. \textbf{Convergence:} we prove
\emph{local} last-iterate convergence at constant anchor strength --- the anchor
scales the base map's rotation by $1-\beta$, crossing the stability boundary and
turning a recurrent base into a contraction (global convergence is conjectured at small
$\beta$; we characterize a large-$\beta$ consensus failure).
\textbf{Robustness:} GARIP matches R-NaD's peak performance --- on matrix games, the
Coin Game, and the board games Connect Four/Othello both moving references are far
more robust than fixed-magnet and magnet-free baselines --- but is the better
\emph{hyperparameter default}; we report it both ways: over the full grid collapse
rates are statistically indistinguishable, yet at conventional parameterizations a
matched-mean-lag setting collapses in $0/40$ vs $10/40$ seeds (a snapshot matches it
only by knowing to shorten $K$). The boundaries: an \emph{anticipatory}
(negative-weight) reference does better still on the stale side, and the advantage
appears only where naive self-play cycles (five deep self-play loops). All experiments
are pure JAX and reproducible.
\end{abstract}

\section{Introduction}

In a two-player zero-sum game with payoff matrix $A$, the row player maximizes
and the column player minimizes $V(x,y)=x^\top A y$. The unique solution
concept is the Nash equilibrium, and the natural learning rule --- each player
ascending/descending its own payoff --- is \emph{self-play gradient ascent}
(SGA). It is a classical fact that SGA cycles: the trajectory orbits the
equilibrium on a closed conservative field and the last iterate never
converges, even though the \emph{time average} does~\cite{golowich2020last}. For
modern deep multi-agent RL, where one keeps the last network and not a time
average over parameters, last-iterate behavior is what matters. The optimization
literature restores last-iterate convergence with \emph{extragradient}
\cite{korpelevich1976} and \emph{optimistic} / past-gradient updates
\cite{popov1980,rakhlin2013optimistic,daskalakis2018optimistic,mertikopoulos2019omd},
and more recently with \emph{Halpern anchoring} toward a reference point for
accelerated last-iterate rates \cite{halpern1967,yoon2021accelerated,cai2022finite}.
GARIP's update is an optimistic step composed with a Halpern anchor whose anchor
\emph{moves} (the running average) rather than being fixed.

A now-standard fix is to add a regularizer that pulls the policy toward a
\emph{reference} (or ``magnet'') policy, breaking the conservative rotation
with a contractive potential. The methods differ in \emph{which} reference they
use and \emph{how it moves}:
\begin{itemize}
  \item \textbf{MMD}~\cite{sokota2023mmd}: a \emph{fixed} reference (e.g.\
  uniform). Converges to the regularized (QRE) equilibrium; reaching true Nash
  needs the regularization annealed to zero.
  \item \textbf{R-NaD}~\cite{perolat2022rnad}: a \emph{periodic snapshot}
  reference, reset to the current policy every $K$ steps. Reaches Nash
  \emph{without} annealing and powers DeepNash on Stratego.
  \item \textbf{NFSP / fictitious play}~\cite{heinrich2016nfsp}: best-respond to
  the average \emph{opponent}; no self-anchor.
\end{itemize}

We examine a fourth point in this design space: anchor to the policy's own
\emph{running average}. We call the resulting method \textbf{GARIP}
(Generative Adversarial Reciprocal Iterative Play); it grew out of a
cycle-consistency reading of zero-sum play (Section~\ref{sec:method}). The
running average is a moving reference, like R-NaD's snapshot, so it too reaches
Nash without annealing --- but it is updated \emph{continuously} rather than in
discrete jumps.

\paragraph{The contribution, made precise.} A moving reference that does not
anneal is R-NaD's idea, and GARIP matches R-NaD's peak performance --- we claim no
more on that axis: at each method's tuned best the two tie across matrix games and
deep RL. GARIP's advantage is hyperparameter robustness, and it is \emph{structural},
not a tuning accident. Held for $K$ steps, R-NaD's snapshot has a sawtooth lag
whose peak ($K$) is twice its mean ($K/2$); collapse is driven by that peak. A
running average has a \emph{flat} lag profile, and we prove (Prop.~\ref{prop:peakmin})
that a flat profile \emph{uniquely minimizes} the peak lag among \emph{causal}
averages at fixed mean lag --- so the running average is the \emph{collapse-optimal
causal} reference shape, and any sawtooth is strictly dominated. Since R-NaD is
parameterized by reset frequency (a mean-like quantity), conventional $K$
under-counts its collapse-driving peak by $2\times$: $\rho\!=\!10^{-2}$ and
$K\!=\!200$ share a mean lag yet collapse $0\%$ versus $25\%$ of seeds
(Sec.~\ref{sec:theory}). On the convergence axis we prove a \emph{local} last-iterate
theorem (Prop.~\ref{prop:localli}): the anchor scales the base map's rotation by
$1-\beta$, crossing the stability boundary so a recurrent base becomes a contraction
--- the tabular form of the same anti-cycling mechanism; the global version is
conjectured at small $\beta$, with a characterized large-$\beta$ consensus failure. The
thesis:
\begin{quote}\itshape
A running-average reference matches R-NaD's peak performance and is the better
\emph{default} for non-annealed self-play: under an empirical collapse-tracks-peak
law it is the peak-minimizing reference \emph{among causal convex averages}, so its
conventional rate is collapse-free where the snapshot's conventional reset period is
not --- which the snapshot recovers only by knowing to shorten $K$. The boundaries:
the aggregate gap is a tie, an anticipatory (negative-weight) reference does better
on the stale side, and the edge appears only where naive self-play cycles.
\end{quote}

\section{The GARIP method}
\label{sec:method}

\paragraph{Cycle-consistency view.} A zero-sum game is the adversarial half of
a CycleGAN: the two players are opposed ``generators'' $G,F$ given by smoothed
best-response maps,
\[
G(x)=\softmax\!\big(\!-(x^\top A)/\tau\big),\quad
F(y)=\softmax\!\big((Ay)/\tau\big),
\]
and the equilibrium is the fixed point of the round trip $F\!\circ\!G$.
``Be a best response to your average'' is exactly cycle-consistency
$F(G(\bar{x}))\approx\bar{x}$. GARIP enforces it by anchoring each update toward
the running average, which is the self-consistent reference.

We are candid that this is \emph{motivation}, not mechanism. The cycle residual
$r_{\text{cyc}}$ (Alg.~\ref{alg:garip}, line~5) is computed only as a convergence
diagnostic and never enters the update --- mechanically, GARIP is optimistic
mirror ascent \cite{daskalakis2018optimistic} composed with a moving Halpern
anchor \cite{halpern1967,yoon2021accelerated}, and it does \emph{not} optimize a
cycle-consistency loss. The CycleGAN lens explains \emph{why} the running average
is the right anchor (it is the $F\!\circ\!G$ fixed point); the name is a
historical artifact of the method's lineage and not a claim that cycle-consistency
is load-bearing.

\paragraph{Matrix-game update.} Maintain iterates $x,y$ and running averages
$\xbar,\ybar$. Each step takes an \emph{optimistic} entropic ascent on the game
value, then a constant-strength Halpern step toward the moving anchor
(Algorithm~\ref{alg:garip}).

\begin{algorithm}[t]
\caption{GARIP (matrix-game form, row player; column $y$ symmetric)}
\label{alg:garip}
\begin{algorithmic}[1]
\State \textbf{Input:} payoff $A$, step $\eta$, anchor weight $\beta$, temperature $\tau$
\State \textbf{Init:} $x \gets$ uniform, $\xbar \gets x$, $g^{\text{prev}} \gets 0$
\For{$t = 1, 2, \dots, T$}
  \Statex \quad$\triangleright$ \emph{CycleGAN generators (smoothed best responses):}
  \State $G(x) \gets \softmax\!\big(\!-(x^\top A)/\tau\big)$ \Comment{opponent's BR to $x$}
  \State $\hat{x} \gets F(G(x)) = \softmax\!\big((A\,G(x))/\tau\big)$ \Comment{cycle round-trip $F\!\circ\!G$}
  \State $r_{\text{cyc}} \gets \lVert \hat{x} - x\rVert$ \Comment{$r_{\text{cyc}}\!=\!0 \Leftrightarrow$ equilibrium}
  \Statex \quad$\triangleright$ \emph{Optimistic ascent, then anchor to the self-consistent reference:}
  \State $g_x \gets A y$ \Comment{game gradient (row maximizes)}
  \State $x_{1/2} \gets \softmax\!\big(\log x + \eta(2 g_x - g^{\text{prev}})\big)$ \Comment{optimistic step}
  \State $x \gets (1-\beta)\, x_{1/2} + \beta\, \xbar$ \Comment{Halpern step, \emph{moving} anchor}
  \State $\xbar \gets \xbar + \tfrac{1}{t+1}(x - \xbar)$ \Comment{$\xbar$: fixed point of $F\!\circ\!G$}
  \State $g^{\text{prev}} \gets g_x$
\EndFor
\State \textbf{return} $x$ (and $\xbar$)
\end{algorithmic}
\end{algorithm}

($y$ symmetric, with descent.) Crucially \textbf{$\beta$ is constant} --- there
is \emph{no} annealing. Setting $\beta=0$ recovers optimistic mirror ascent;
replacing $\xbar$ by a periodic snapshot recovers R-NaD; replacing it by a fixed
distribution recovers MMD. This makes the four methods a clean ablation over
\emph{which reference, moving how}.

\paragraph{Deep-RL form.} In function-approximation settings the same idea is
PPO self-play with a $\lambda\,\KL(\pi_\theta \,\|\, \pi_{\text{mag}})$ term,
where the magnet $\pi_{\text{mag}}$ is the running-average (Polyak) policy and
the opponent is the average opponent. MMD uses a fixed magnet; R-NaD uses a
periodic-snapshot magnet reset every $K$ updates. The five methods span the
\emph{opponent $\times$ magnet} grid; note that GARIP (average opponent, moving
magnet) and R-NaD (current opponent, snapshot magnet) differ on \emph{both} axes,
so their head-to-head gap is not a clean single-variable isolation of the
reference. We therefore fill the full $2\times2$ grid at the standard config
($\lambda\!=\!0.5$, $K\!=\!200$, $\rho\!=\!10^{-2}$, 10 seeds; exploit return,
lower is more robust): GARIP (avg+moving) $-11.3$, R-NaD (current+snapshot)
$-11.1$, average+snapshot $-12.3$, current+moving $-11.9$, all with $0\%$
collapse and std $\approx\!1.5$. \emph{Every} single-axis effect (changing the
magnet, or changing the opponent) is below $1.3$ --- smaller than the seed noise.
So at default settings neither the magnet nor the opponent axis matters: the
GARIP--R-NaD tie is not driven by either variable. The reference choice only
becomes decisive in the \emph{staleness} regime (Sec.~\ref{sec:exp}), which is
the entire content of the robustness claim.

\section{Staleness: the failure mode of every moving reference}
\label{sec:theory}

\paragraph{Two averagers.} GARIP realizes the running average two ways, and the
distinction is load-bearing. In the matrix games (Alg.~\ref{alg:garip}, line~10)
it is the \emph{cumulative} mean, weight $1/(t{+}1)$ on the newest iterate; in
the deep-RL setting it is a fixed-rate \emph{exponential moving average} (Polyak),
$\bar\theta \leftarrow (1{-}\rho)\bar\theta + \rho\,\theta$. These have different
lag profiles, and the $D/\rho$ bound below governs only the EMA.

\paragraph{One mechanism, three references.} Let $D$ be the per-step policy
drift and let a reference lag the current policy by $\ell$ steps. Under KL weight
$\lambda$ it exerts a restoring force $\sim\!\lambda\,\ell\,D$ toward an outdated
target; once this overwhelms the game gradient the policy is dragged to the stale
reference and \emph{collapses to an exploitable policy}. The references differ
only in $\ell$:
\begin{itemize}
  \item \textbf{periodic snapshot} (R-NaD): held fixed for $K$ steps,
  $\ell=\Theta(K)$, stale force $\sim\!\lambda K$;
  \item \textbf{fixed-rate EMA} (deep-RL GARIP): $\ell=\Theta(1/\rho)$, stale
  force $\sim\!\lambda/\rho$;
  \item \textbf{cumulative mean} (matrix GARIP): $\ell=\Theta(t)$ grows without
  bound, \emph{but} $D$ vanishes as the iterates approach equilibrium, so the
  product stays bounded and the mean catches up --- which is why the matrix form
  has no collapse region.
\end{itemize}
Consequently \textbf{the EMA used in deep RL has its own staleness axis.} A slow average ($\rho\!\to\!0$) at non-trivial $\lambda$ is exactly as
stale as a long snapshot --- the correspondence $\lambda K \leftrightarrow
\lambda/\rho$ is symmetric. So a running average does \emph{not} ``cannot go
stale'': our own sweep (Sec.~\ref{sec:exp}) finds GARIP collapses at small $\rho$
($\rho\!\approx\!10^{-3}$) for $\lambda\!\ge\!0.5$, just as R-NaD collapses at
large $K$.

\paragraph{The corrected prediction.} This gives a sharper, falsifiable claim
than ``GARIP never collapses'': \emph{each method degrades monotonically in its
own staleness product} ($\lambda K$ for R-NaD, $\lambda/\rho$ for EMA-GARIP), and
the practical question is which method's \emph{commonly used} hyperparameter
range sits farther from its collapse boundary. R-NaD's default reset period is
$O(10^2\text{--}10^3)$, inside the regime our sweep finds dangerous; GARIP's
default $\rho=10^{-2}$ is an order of magnitude from its boundary
$\rho\!\approx\!10^{-3}$. Section~\ref{sec:exp} measures both boundaries on a
matched grid (GARIP's $\rho$ swept down to $1.25\!\times\!10^{-3}$, symmetric to
R-NaD's $K\!=\!800$) and reports the collapse rates with Wilson confidence
intervals. The prediction holds: GARIP's collapse boundary is at
$\rho\!\approx\!0.005$ (collapse-free for $\rho\!\ge\!0.01$), R-NaD's at
$K\!\approx\!100$--$200$ (collapsing for the larger $K$ used at scale) ---
symmetric mechanisms, but GARIP's boundary sits past its default rate and
R-NaD's inside common reset periods (Sec.~\ref{sec:exp}). A fuller treatment is
in the supplementary \texttt{staleness\_analysis.md}.

\paragraph{Why it is structural: the running average is the optimal reference
shape.} Model a moving reference as a convex combination of past iterates,
$r_t=\sum_{s\le t}w_{t,s}\theta_s$ ($w_{t,s}\!\ge\!0$, $\sum_s w_{t,s}\!=\!1$). Under
local drift $\lVert\theta_{t+1}-\theta_t\rVert\!\approx\!D$, its lag is
$\ell_t=\lVert\theta_t-r_t\rVert/D=\sum_s w_{t,s}(t-s)$, the weighted mean age, and
it exerts a stale force $\sim\!\lambda\,\ell_t\,D$. The load-bearing physics is
\emph{empirical}: collapse tracks the \emph{peak} force --- the GARIP and R-NaD
\emph{peak-lag} curves overlay (Fig.~\ref{fig:lag}), so a flat reference and a
snapshot peaking at the same lag collapse identically and the reset
\emph{discontinuity} contributes nothing; only the peak magnitude matters. Granting
that empirical law, which reference \emph{shape} minimizes the peak is a one-line
theorem.

\begin{proposition}[The running average minimizes peak lag among causal averages]
\label{prop:peakmin}
Among all \emph{causal convex} references $r_t=\sum_{s\le t}w_{t,s}\theta_s$
($w_{t,s}\!\ge\!0,\ \sum_s w_{t,s}\!=\!1$) with a given time-averaged lag $\bar\ell$,
the peak lag obeys $\ell_{\max}\ge\bar\ell$, with equality iff the lag profile
$\{\ell_t\}$ is constant in $t$. Hence, at comparable per-step drift $D$ (so the
collapse-safe weight scales as $\lambda^\star\!\propto\!1/(\ell_{\max}D)$), the
constant-lag profile is the unique safest \emph{shape} at matched mean lag. The
running average realizes it --- the EMA's lag is stationary at
$(1-\rho)/\rho\!\approx\!1/\rho$ --- whereas the periodic snapshot is a sawtooth with
$\ell_{\max}=2\bar\ell$, strictly dominated by a factor $2$.
\end{proposition}
\noindent\emph{(Proof in App.~\ref{app:proofs}.)}

\begin{remark}[Scope: non-anticipatory averages]
\label{rem:causal}
Optimality is \emph{within} causal convex averages. Extrapolating references
(negative/Anderson-style weights) admit $\ell\!<\!0$ and can lower the peak force
further --- and indeed an \emph{anticipatory} double-EMA magnet measurably widens
the safe basin on the stale side (App.~\ref{sec:extrap}), confirming the running
average is optimal only in the causal \emph{convex} class (the double EMA is still
causal --- only past iterates --- but not a convex combination). It is not free: the
same anticipation overshoots at fast averaging rates, so it complements rather than
dominates GARIP.
\end{remark}

\noindent So, granting the empirical peak-force law, the running average is the
peak-minimizing reference shape among causal averages, and a sawtooth's peak is
twice its mean. Whether this matters depends on parameterization. R-NaD is
parameterized by reset frequency $K$, a mean-like quantity, so a $K$ chosen to feel
``moderately stale'' carries a collapse-driving peak twice as large; the EMA's peak
equals its mean. At matched \emph{peak} lag the two tie --- one could set
$K\!\approx\!100$ and recover safety --- so the gap is a property of the conventional
parameterization, not of the methods in the abstract. Concretely, at the standard
$\rho\!=\!10^{-2}$ and $K\!=\!200$ (same mean lag $100$), GARIP collapses in $0/40$
seeds (including the stress $\lambda\!=\!2$) versus R-NaD's $10/40$. The lone outlier in
Fig.~\ref{fig:lag}, an extreme $\rho\!=\!1.25\!\times\!10^{-3}$, is a never-used EMA
so slow its \emph{permanent} lag exceeds the snapshot's transient peak.

\begin{figure}[tb]
\centering
\includegraphics[width=0.95\linewidth]{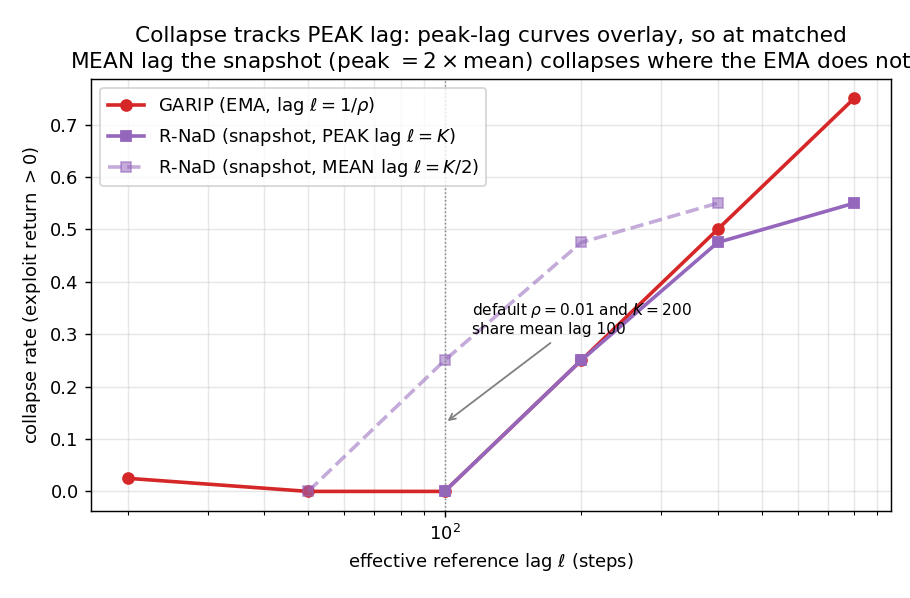}
\caption{Collapse rate vs.\ effective reference lag (deep-RL sweep, 10 seeds). The
EMA's lag is flat ($1/\rho$); the snapshot's is a sawtooth (peak $K$, mean $K/2$).
Each point averages the \emph{same} four $\lambda$, so the overlay is not a
$\lambda$ artifact. The GARIP and R-NaD \emph{peak}-lag curves overlay at the
matched points ($\ell\!=\!100,200,400$) --- collapse is peak-driven, and the reset
discontinuity adds nothing --- so at matched \emph{mean} lag R-NaD (dashed)
collapses earlier. $\rho\!=\!10^{-2}$ and $K\!=\!200$ share mean lag $100$ yet
collapse $0\%$ vs $25\%$.}
\label{fig:lag}
\end{figure}

\section{Tabular dynamics: fixed points and the unbiased anchor}
\label{sec:analysis}

The robustness theorem (Prop.~\ref{prop:peakmin}) is the load-bearing theory; this
section adds the matrix-game dynamics that complete the picture --- why GARIP
reaches Nash \emph{without annealing}. On the convergence axis it gives a
\emph{local} last-iterate theorem (Prop.~\ref{prop:localli}: the anchor crosses the
stability boundary) and leaves only the \emph{global} statement conjectured; we are
explicit about that boundary.
On a matrix game $A$ the row plays $x\in\Delta_m$ against $y\in\Delta_n$ for value
$x^\top\! A y$. GARIP's row update (column symmetric), with step $\eta$ and constant
anchor $\beta\in(0,1)$, is
\begin{align*}
g_t &= A y_t,\\
x_{t+\frac12} &= \softmax\!\big(\log x_t + \eta(2g_t - g_{t-1})\big),\\
x_{t+1} &= (1-\beta)\,x_{t+\frac12} + \beta\,\bar x_t,\\
\bar x_{t+1} &= \bar x_t + \tfrac{1}{t+2}\big(x_{t+1}-\bar x_t\big).
\end{align*}
The half-step is optimistic multiplicative weights (OMWU); the anchor is a Halpern
step toward the running average $\bar x_t$. Two facts are exact, and together they
explain why GARIP reaches Nash without annealing.

\begin{lemma}[The anchor contracts deviation from the average]
\label{lem:contract}
For every $t$, $\;x_{t+1}-\bar x_t=(1-\beta)\,(x_{t+\frac12}-\bar x_t)$, hence
$\lVert x_{t+1}-\bar x_t\rVert = (1-\beta)\,\lVert x_{t+\frac12}-\bar x_t\rVert$.
\end{lemma}
\noindent The anchor shrinks the gap to the \emph{instantaneous} anchor $\bar x_t$
by exactly $1-\beta$ each step (the running average itself then moves by
$\bar x_{t+1}-\bar x_t=O(1/t)$, which the identity does not net out --- contraction
is toward the current anchor, not a fixed point). This is the oscillation damping of
Sec.~\ref{sec:theory} made exact, and it carries \emph{no} fixed-point bias:

\begin{proposition}[Fixed points are Nash; the anchor is unbiased]
\label{prop:fp}
(i) Every Nash equilibrium of $A$ is a fixed point of GARIP, with the anchor
inactive there. (ii) Conversely, every \emph{interior} fixed point is a Nash
equilibrium. (iii) Unlike MMD, whose \emph{fixed} magnet $m$ exerts a nonzero pull
at Nash --- making its fixed point the magnet-regularized (quantal-response)
equilibrium, $O(\eta)$ off Nash --- GARIP's self-referential anchor vanishes at any
fixed point, so no annealing is needed to remove a magnet bias.
\end{proposition}
\noindent\emph{(Proof in App.~\ref{app:proofs}.)}

\begin{remark}[The $\beta\!\to\!0$ core]
\label{rem:omwu}
At $\beta=0$ GARIP is exactly OMWU, which converges to Nash in the \emph{last
iterate} at a linear rate in bilinear games \cite{wei2021linear,daskalakis2018optimistic}.
The anchor is the only addition.
\end{remark}

We can now prove last-iterate convergence \emph{locally} --- and, more importantly,
show analytically that the anchor is what crosses the stability boundary.

\begin{proposition}[Local last-iterate convergence; the anchor crosses the
stability boundary]
\label{prop:localli}
Let $x^\star$ be an interior Nash of an $m\times n$ zero-sum game. The base half-step
linearizes there into independent modes; let $M=\mu\,e^{i\omega}$ be a mode (an
eigenvalue of the base operator). For the un-anchored base ($\beta\!=\!0$),
non-optimistic MWU is a rotation with $\mu\ge1$ --- its last iterate recurs
(orbits/spirals out) and does \emph{not} converge --- while optimism gives $\mu<1$.
Because the anchor and the running average act \emph{identically on every coordinate}
(scalar operators), they commute with the base operator; in the basis diagonalizing
it the full linearized GARIP map block-diagonalizes, per mode, into
\[
\begin{gathered}
J(M)=\begin{pmatrix}(1-\beta)M & \beta\\[1pt] \rho(1-\beta)M & 1-\rho(1-\beta)\end{pmatrix},\\[3pt]
\det J=(1-\rho)(1-\beta)M,
\end{gathered}
\]
so $\rho(\text{full})=\max_k\rho\big(J(M_k)\big)$ over the modes $M_k$. Each block has a
\emph{fast} policy eigenvalue of modulus $(1-\beta)\mu$ and a \emph{slow} average
eigenvalue near $1$; for $\beta\in(0,1)$, $\rho\in(0,1]$ and small $\eta$ both lie
strictly inside the unit disk, so $x^\star$ is locally asymptotically stable in the
last iterate --- with \emph{no annealing}. Crucially the anchor scales each base mode
by exactly $1-\beta$ (Lemma~\ref{lem:contract}): it turns a recurrent rotation
($\mu\!=\!1\Rightarrow(1-\beta)\mu<1$) into a contraction, so it \emph{induces}
last-iterate convergence for the non-optimistic base --- the tabular form of the
deep-RL collapse-prevention mechanism (\S\ref{sec:theory}).
\end{proposition}
\noindent\emph{(Proof in App.~\ref{app:proofs}: per-mode block-diagonalization ---
exact for the one-step base on any $m\times n$ game, confirmed numerically --- plus the
first-order slow-mode expansion, whose coefficient is complex because $M$ is a
rotation.)}

\noindent What \emph{stays} open is the global statement. The right potential-level
tool is exact:

\begin{lemma}[Potential-level anchor contraction]
\label{lem:klconv}
Write $\Phi_t=\KL(x^\star\Vert x_t)$, $\Phi_{t+\frac12}=\KL(x^\star\Vert x_{t+\frac12})$,
$\Psi_t=\KL(x^\star\Vert\bar x_t)$. Because $q\mapsto\KL(x^\star\Vert q)$ is convex,
\[
\Phi_{t+1}\le(1-\beta)\,\Phi_{t+\frac12}+\beta\,\Psi_t,
\qquad \Psi_t\le\tfrac1t\textstyle\sum_{s<t}\Phi_s .
\]
\end{lemma}
\noindent\emph{(Proof in App.~\ref{app:proofs}, by convexity of $\KL$.)}

\begin{conjecture}[Global last iterate, at small $\beta$ --- with a failure mode at
large $\beta$]
\label{conj:li}
For \emph{small enough} $\beta$ and $\eta$, GARIP's last iterate converges to Nash from
any interior start. \emph{Reduction (two open steps collapse to one).} If the OMWU
half-step is potential-non-increasing, $\Phi_{t+\frac12}\le\Phi_t-D_t$ ($D_t\ge0$;
\cite{wei2021linear}), Lemma~\ref{lem:klconv} gives
$\Phi_{t+1}\le(1-\beta)\Phi_t+\beta\Psi_t-(1-\beta)D_t$, so since $1-\beta<1$,
\[
\Psi_t=\KL(x^\star\Vert\bar x_t)\to0\ \Longrightarrow\ \Phi_t\to0 :
\]
global last iterate follows from a \emph{single} estimate, $\bar x_t\to$ Nash.
\emph{But that estimate is exactly where the constant anchor can fail, and we found it
does.} In long-horizon sweeps (many random and near-boundary starts; repo script), at the
default small $\beta$ the last iterate descends to Nash on every game we tried --- but
slowly on large games (a Ces\`aro-limited $\sim t^{-1/3}$ rate, not acceleration).

\emph{The premature-consensus regime.} At larger $\beta$ on larger interior-Nash games
the anchor instead drives $x_t$ and $\bar x_t$ into \emph{premature consensus}
($\Vert x_t-\bar x_t\Vert\to0$ with exploitability frozen at $\sim10^{-1}$, so
$\bar x_t\not\to$ Nash), and the iterate stalls. The boundary is orderly --- the critical
anchor $\beta^\star(d)$ shrinks with game size, and the default $\beta\!=\!0.02$ sits
below it through $d\!\approx\!10$ (App.~\ref{app:phase}). Pure OMWU ($\beta\!=\!0$)
reaches Nash to $\sim10^{-7}$ on these games, so on an already-convergent base a strong
anchor over-regularizes --- the tabular face of the Animal-Shogi lesson
(\S\ref{sec:boardgames}). This is a \emph{basin-of-attraction} failure, not an undercut
of no-annealing: Prop.~\ref{prop:fp} (fixed points are Nash, anchor unbiased) is about
\emph{where} fixed points sit and is untouched; premature consensus is about \emph{which}
basin a large anchor steers into. The recommended \emph{deep-RL} point is unaffected ---
there collapse is reference \emph{staleness}, not anchor \emph{strength}, and the anchor
is load-bearing. A global proof for small $\beta$ and a sharp $\beta^\star(d)$ law are
open.
\end{conjecture}

\noindent The collapse law has a derivable \emph{discrete} (not continuous) origin --- a
linear one-mode map overshoots through the peak-lag state while the continuous-time
delay only damps; we defer the derivation to App.~\ref{app:discrete}.

\section{Experiments}
\label{sec:exp}

All code is pure JAX; matrix and poker exploitability are \emph{exact}
(LP-free NashConv); the deep-RL proxy freezes the policy and trains a fresh
best-response PPO agent. Unless noted, results are over 10 seeds.

\subsection{Matrix games: last-iterate convergence}

\begin{table*}[tp]
\centering
\caption{Matrix games: last-iterate exploitability (10 seeds, 5000 steps,
\emph{no annealing}). Lower is better; bold marks the best method on the
non-trivial random games. GARIP and R-NaD both reach near-Nash; MMD plateaus at
the QRE; SGA never converges.}
\label{tab:matrix}
\begin{tabular}{lcccccc}
\toprule
game & GARIP (ours) & R-NaD & MMD & Optimistic MD & Fictitious play & Self-play GA \\
\midrule
rps                 & 0.0003 & 0.0012 & 0.0000 & 0.0004 & 0.0247 & 1.999 \\
matching pennies    & 0.0001 & 0.0000 & 0.0000 & 0.0000 & 0.0200 & 1.998 \\
random $10\times10$ & \textbf{0.0028} & 0.0944 & 0.1207 & 0.0634 & 0.0283 & 2.522 \\
random $12\times6$  & \textbf{0.0052} & 0.0617 & 0.1400 & 0.0250 & 0.0169 & 2.101 \\
\bottomrule
\end{tabular}
\end{table*}

Table~\ref{tab:matrix} shows that at each method's \emph{default} configuration
GARIP reaches near-Nash on every game, R-NaD's default trails on the random
games (and only matches GARIP once $(\alpha,K)$ is tuned), and MMD plateaus at
the regularized equilibrium. The exploitability curves and the RPS simplex
trajectory (App.~\ref{app:poker}) visualize the mechanism --- SGA spirals outward
forever while GARIP walks straight into Nash and stays.

\subsection{Poker: Kuhn and Leduc}
On two imperfect-information benchmarks with \emph{exact} exploitability --- neural
Kuhn (MLP through the differentiable tree) and tabular Leduc (288 infosets, exact
CFR/BR) --- GARIP is the best regularized method at default tuning: Kuhn $0.076$ vs
R-NaD $0.30$ and MMD $0.54$ (approaching exact-BR fictitious play, $0.043$), and
Leduc $0.074$ vs R-NaD $0.23$ and MMD $0.55$, below CFR's $0.116$ at a matched
$3000$-iteration budget (a reference point, not a converged baseline). These are
\emph{default}-vs-default configs --- we did not separately tune R-NaD's
$(\alpha,K)$ on poker, and a tuned R-NaD would likely close much of the gap, exactly
as on the matrix games --- so this is consistent with the tie-at-each-tuned-best
finding, not a tuned-best poker win. Full tables and curves: App.~\ref{app:poker}.

\subsection{Deep-RL self-play: Coin Game}

We leave exact-solver territory with a deep-RL self-play loop on the Coin Game,
a canonical 2-player testbed vendored from JaxMARL (dependencies stripped to
avoid version coupling) and run zero-sum. A shared PPO actor-critic plays both
sides; the proxy exploitability freezes the policy and trains a fresh
best-response agent (lower / more negative is more robust).
Table~\ref{tab:coin} shows GARIP ($-11.29$) and R-NaD ($-11.07$) \emph{tie} and
are far more robust than the fixed magnet (MMD, $+7.83$) and the magnet-free
baselines (Figure~\ref{fig:coin}). The \emph{moving} reference matters; the
specific moving reference does not, \emph{for peak performance}.

\begin{table}[tbp]
\centering
\caption{Coin Game deep-RL self-play: proxy exploit return (10 seeds; lower is
more robust) at the default config. The five methods span the
opponent$\times$magnet grid; at this config the four magnet/opponent combinations
tie (Sec.~\ref{sec:method}), so the gap to MMD/naive is the magnet's
\emph{presence}, not its type.}
\label{tab:coin}
\small
\begin{tabular}{llc}
\toprule
method & opponent / magnet & exploit return \\
\midrule
\textbf{GARIP}        & avg / running-avg     & $\mathbf{-11.29 \pm 1.51}$ \\
\textbf{R-NaD}        & current / snapshot    & $\mathbf{-11.07 \pm 1.49}$ \\
naive self-play       & current / none        & $-1.72 \pm 2.22$ \\
fictitious self-play  & avg / none            & $-0.04 \pm 2.84$ \\
MMD                   & current / fixed       & $+7.83 \pm 0.93$ \\
\bottomrule
\end{tabular}
\end{table}

\begin{figure}[tbp]
\centering
\includegraphics[width=\linewidth]{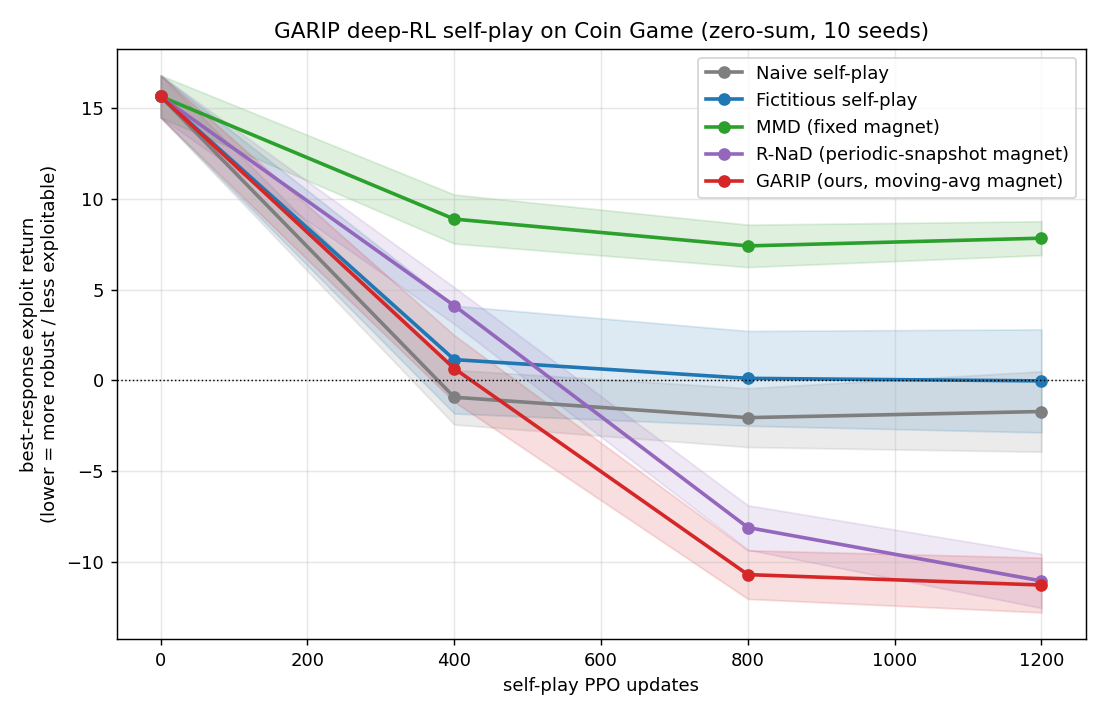}
\caption{Coin Game proxy exploitability vs.\ training (10 seeds). GARIP and
R-NaD reach the lowest (most robust) exploit return; MMD's fixed magnet is
worst.}
\label{fig:coin}
\end{figure}

\subsection{Four more deep environments: board games}
\label{sec:boardgames}

To test whether the deep-RL finding is Coin-Game-specific, we add four
structurally unrelated deep envs: the turn-based, perfect-information, strictly
zero-sum board games \textbf{Connect Four, Othello, Animal Shogi} and \textbf{Hex}
(via \texttt{pgx}). Board games sidestep the proxy metric's main weakness --- a
policy cannot look robust by being passive, because every game terminates with a
win/lose/draw, and the action distribution \emph{is} the strategy (so the magnet
regularizes exactly what is scored). A shared actor-critic plays both colours (the
observation is egocentric); the same opponent$\times$magnet ablation applies. We
reduce turn-based self-play to a single-agent problem from the learner's view
(learner moves, then the opponent replies), and measure robustness by the
\textbf{win-rate of a freshly trained best-responder} against the frozen policy
(lower $=$ harder to exploit; stochastic play with random openings so the rate
grades exploitability). 8 seeds; $1000$ self-play $+\,400$ BR updates
(Hex: $10{,}000$ self-play).

These four games \emph{map the scope} of the advantage and confirm a prediction the
mechanism makes --- a moving reference helps exactly where naive self-play cycles
(Table~\ref{tab:board}, Figure~\ref{fig:board}):
\begin{itemize}
  \item \textbf{Connect Four and Othello replicate the Coin Game ordering cleanly.}
  The moving-reference methods are far more robust than the fixed-magnet and
  magnet-free baselines --- on Connect Four GARIP ($0.43$) and R-NaD ($0.44$) tie
  below $0.5$ (a budget-matched exploiter \emph{loses} to them) while the rest are
  exploited $0.76$--$0.86$; on Othello GARIP ($0.55\!\pm\!0.06$) significantly edges
  R-NaD ($0.67\!\pm\!0.07$; 8 seeds, $p\!<\!0.01$), both far below $0.86$--$0.94$.
  \item \textbf{Animal Shogi is the predicted control}, and we confirmed the
  mechanism directly rather than inferring it. Here \emph{naive} self-play is
  hardest to exploit ($0.25$), exactly as the cycling account predicts for a game
  that does not cycle, and two diagnostics establish it. (i) Its
  exploitability \emph{decreases monotonically} over training
  ($0.52\!\to\!0.35\!\to\!0.28$ at $250/500/1000$ updates) --- it \emph{converges}
  rather than cycles (on Connect Four naive instead stays at $0.79$). (ii) It
  \emph{beats} the regularized policies head-to-head (naive$\,>\,$GARIP
  $0.59/0.40$, naive$\,>\,$R-NaD $0.67/0.32$; each over $n\!=\!2048$ games, so
  $\pm0.02$ at $95\%$), so it is the genuinely stronger player. This small
  ($3\times4$) piece-drop game does not induce the last-iterate cycling that
  regularization cures, so the magnet only handicaps playing strength. Even so,
  GARIP remains the best regularizer: GARIP$\,>\,$R-NaD head-to-head ($0.60/0.39$)
  and in exploitability ($0.43<0.52<0.87$ for GARIP/R-NaD/MMD).
  \item \textbf{Hex marks the learner's ceiling, not a counterexample.} A
  $10{,}000$-update stress test establishes it: \emph{every} method stays
  $0.89$--$0.96$ exploitable because an $11\times11$ connection game is beyond what
  this search-free PPO setup learns robustly, so the metric has nothing to
  separate (R-NaD/GARIP are marginally lowest). This bounds the \emph{learner}, and
  the obvious next step is a stronger one (search/AlphaZero-style).
\end{itemize}
This is a \emph{mechanism test}, not a self-correlation: the diagnostic is the
\emph{shape} of naive's training curve --- the learning dynamics, causally prior to
any endpoint. Naive's exploitability \emph{descends monotonically} on Animal Shogi
($0.52\!\to\!0.35\!\to\!0.28$, each step lower --- converging dynamics) but is
\emph{flat or stuck} on Connect Four ($0.90\!\to\!0.82$), Othello
($0.97\!\to\!0.93$) and Hex ($0.93\!\to\!0.93$). Read Table~\ref{tab:synthesis} as
\emph{dynamics (cause) $\to$ advantage (effect)}: \textbf{the moving reference helps
iff naive's dynamics fail to converge \emph{and} the game is learnable}, and all six
environments comply --- the single game whose naive dynamics descend (Animal Shogi)
is the single game where regularization does not help. GARIP is the best regularizer
wherever the comparison is meaningful.

\begin{table}[tbp]
\centering
\caption{Dynamics (cause) $\to$ advantage (effect). The \emph{shape} of naive
self-play's training curve, and whether the moving-reference advantage then appears.
The single environment whose naive dynamics descend is the single one without an
advantage. (Endpoint metric is per-env: last-iterate exploitability for
matrix/Coin Game, best-response win-rate for board games; lower $=$ more robust.)}
\label{tab:synthesis}
\small
\resizebox{\columnwidth}{!}{%
\begin{tabular}{llcc}
\toprule
environment & naive dynamics & best$_{\text{reg}}$ vs naive & advantage \\
\midrule
matrix (random) & orbits                        & $0.003$ vs $2.0$  & yes \\
Coin Game       & oscillates                    & $-11.3$ vs $-1.7$ & yes \\
Connect Four    & flat ($0.90\!\to\!0.82$)      & $0.43$ vs $0.79$  & yes \\
Othello         & flat ($0.97\!\to\!0.93$)      & $0.55$ vs $0.94$  & yes \\
Animal Shogi    & \textbf{descends} ($0.52\!\to\!0.28$) & naive wins ($0.25$) & no \\
Hex             & stuck (unlearnable)           & all $\sim\!0.9$   & no \\
\bottomrule
\end{tabular}}
\end{table}

\begin{table}[tbp]
\centering
\caption{Board-game self-play robustness: best-response win-rate vs.\ the frozen
policy (mean$_{\pm\text{std}}$ over 8 seeds; lower $=$ harder to exploit). Bold $=$
most robust per game.}
\label{tab:board}
\small
\setlength{\tabcolsep}{4pt}
\resizebox{\columnwidth}{!}{%
\begin{tabular}{lccccc}
\toprule
game & GARIP & R-NaD & MMD & naive & fict. \\
\midrule
Connect Four & $\mathbf{0.43_{\pm.06}}$ & $0.44_{\pm.05}$ & $0.76_{\pm.02}$ & $0.79_{\pm.09}$ & $0.86_{\pm.05}$ \\
Othello      & $\mathbf{0.55_{\pm.06}}$ & $0.67_{\pm.07}$ & $0.92_{\pm.01}$ & $0.94_{\pm.02}$ & $0.86_{\pm.06}$ \\
Animal Shogi & $0.43_{\pm.04}$ & $0.52_{\pm.06}$ & $0.87_{\pm.01}$ & $\mathbf{0.25_{\pm.07}}$ & $0.58_{\pm.12}$ \\
Hex (10k)    & $0.92_{\pm.04}$ & $\mathbf{0.89_{\pm.03}}$ & $0.96_{\pm.01}$ & $0.93_{\pm.03}$ & $0.96_{\pm.01}$ \\
\bottomrule
\end{tabular}}
\end{table}

\begin{figure}[tbp]
\centering
\includegraphics[width=\linewidth]{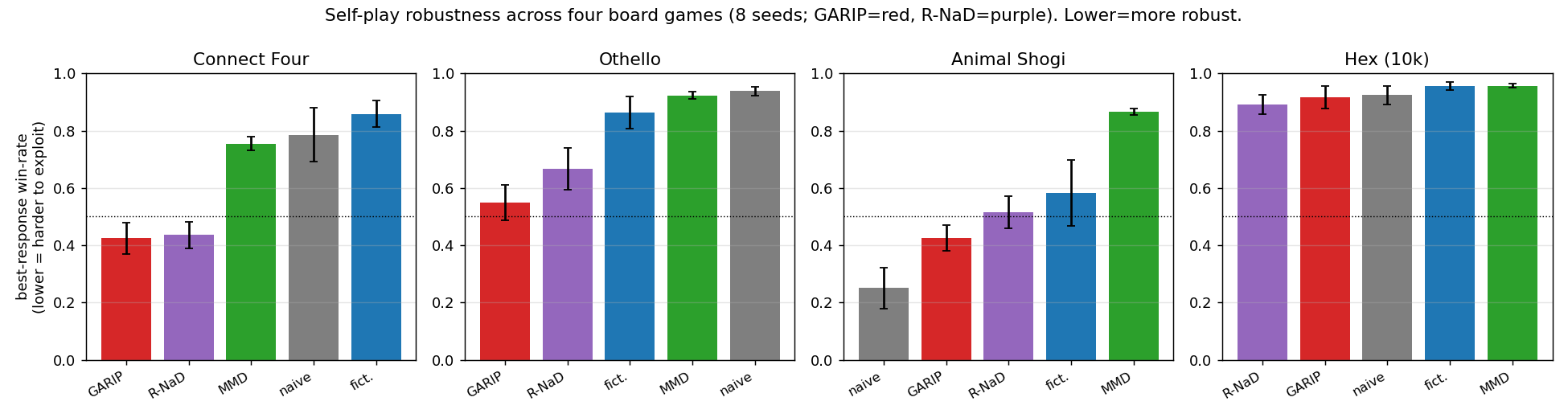}
\caption{Best-response win-rate across four board games (8 seeds; GARIP red, R-NaD
purple; lower $=$ more robust). Moving-reference methods dominate on Connect Four
and Othello; Animal Shogi does not cycle (naive wins, but GARIP is the best
regularizer); Hex saturates (all $\sim$0.9, under-trained).}
\label{fig:board}
\end{figure}

\subsection{Hyperparameter robustness}
\label{sec:robust}

We now sweep each method's hyperparameter grid and score \emph{every}
configuration. On a 5$\times$5 matrix-game grid (App.~\ref{app:sweeps}), GARIP's
median exploitability is $0.031$ vs R-NaD's $0.332$, and 56\% of GARIP configs
converge vs R-NaD's 20\%. R-NaD only works in a narrow band $\alpha\!\le\!0.5$ and
is sensitive to its reset period $K$; GARIP has no reset period and a much wider
basin.

\paragraph{And in deep RL --- a basin, not immunity.} We sweep the Coin Game on a
matched grid: GARIP over $\lambda\!\in\!\{0.25,0.5,1,2\}\times\rho\!\in\!\{1.25,2.5,
5,10,20,50\}\!\times\!10^{-3}$ (the two slowest $\rho$ push the EMA lag
$\sim\!1/\rho$ \emph{into GARIP's own stale regime}, symmetric to R-NaD's
$K\!=\!800$), R-NaD over $\lambda\times K\!\in\!\{100,200,400,800\}$, 10 seeds.
Over the full grids the aggregate collapse rates are \emph{statistically
indistinguishable} --- GARIP $25.4\%$ $[20.3,31.3]$, R-NaD $31.9\%$ $[25.2,39.4]$
(Wilson 95\% CIs) --- precisely because GARIP's grid now contains configs that
collapse. The advantage is \emph{where} each collapse boundary sits relative to
the rate one actually uses (Table~\ref{tab:sens-coin},
Figure~\ref{fig:sens-coin}): GARIP collapses only at slow averaging
($\rho\!\le\!0.005$) and is \textbf{essentially collapse-free across
$\rho\!\ge\!10^{-2}$} ($1/120$ runs --- the single failure is one seed at the
\emph{fastest} rate $\rho\!=\!0.05$, i.e.\ too-fast tracking, not staleness),
whereas R-NaD already collapses $25\%$ at the commonly-used $K\!=\!200$ and $55\%$
at $K\!=\!800$, safe only at the short $K\!=\!100$. GARIP's failure region thus lies an order of magnitude slower than
the standard Polyak rate; R-NaD's overlaps the reset periods used at scale. This
is the basin-width advantage --- \emph{not} collapse-immunity (an earlier draft's
$7\%$ figure was a grid artifact from a $\rho$ range that stopped short of the
stale regime, corrected here).

\paragraph{Exact-metric confirmation on Leduc.} Finally, the staleness
prediction --- R-NaD degrades monotonically in $\alpha K$ --- is confirmed on a
\emph{second} environment with an \emph{exact} metric (App.~\ref{app:sweeps}).
R-NaD's exploitability rises monotonically from 0.24 (low $\alpha$, low $K$) to
2.01 (high $\alpha$, high $K$) and \emph{never} reaches near-Nash on Leduc (no
config below 0.1), because the stale magnet leaves a residual bias even at its
best. GARIP reaches 0.04.

\medskip
\begin{center}
\small
\captionof{table}{Deep-RL collapse rate (exploit return $>0$) vs.\ staleness
depth (Coin Game, 10 seeds, 40 runs per column entry). Both methods collapse at
their stale end; GARIP's boundary lies \emph{past} its default $\rho\!=\!10^{-2}$,
R-NaD's lies \emph{inside} commonly-used $K$. Aggregate rates over the full grids
are indistinguishable (GARIP $25\%$, R-NaD $32\%$); the boundary is the story.}
\label{tab:sens-coin}
\begin{tabular}{cc@{\quad}cc}
\toprule
GARIP $\rho$ & collapse & R-NaD $K$ & collapse \\
\midrule
0.00125          & 0.75          & 100  & \textbf{0.00} \\
0.0025           & 0.50          & 200  & 0.25 \\
0.005            & 0.25          & 400  & 0.47 \\
\textbf{0.01}$^\dagger$ & \textbf{0.00} & 800  & 0.55 \\
0.02             & 0.00          & ---  & --- \\
0.05             & 0.03          & ---  & --- \\
\bottomrule
\end{tabular}
\par\smallskip {\footnotesize $^\dagger$ GARIP's default Polyak rate.}

\medskip
\includegraphics[width=\linewidth]{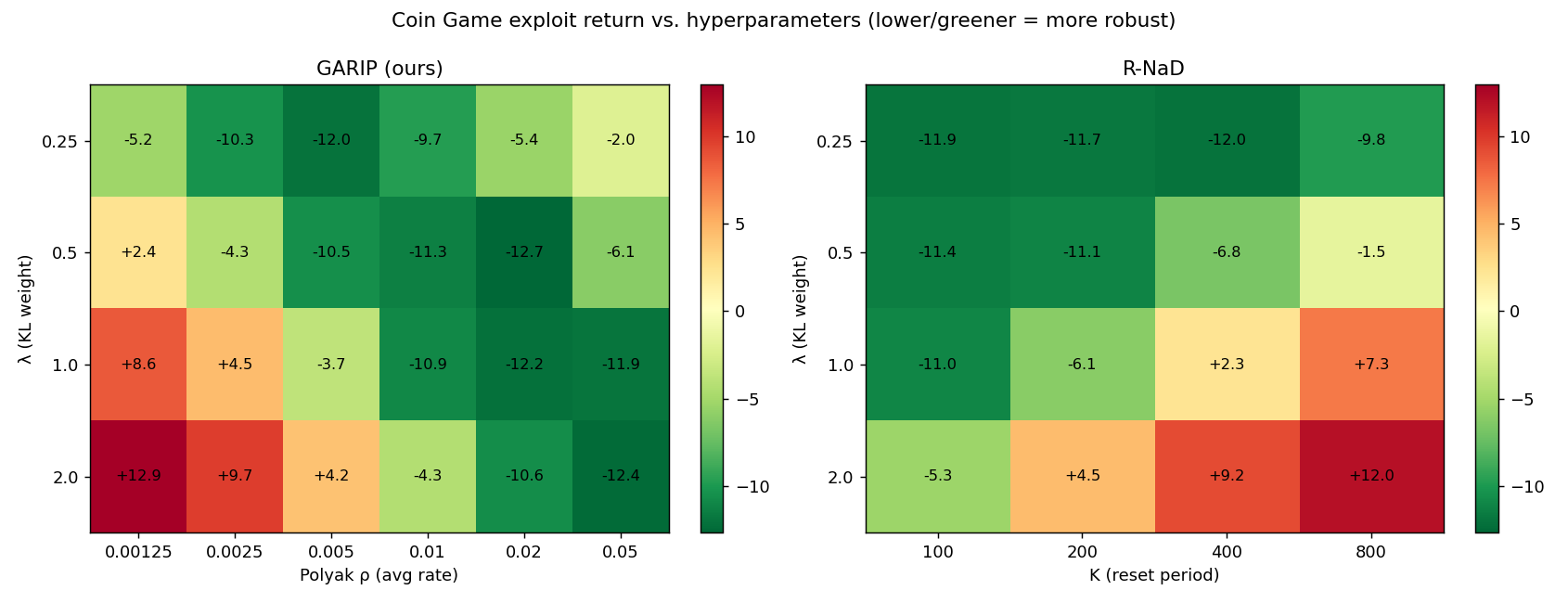}
\captionof{figure}{Deep-RL hyperparameter sensitivity (10 seeds; GARIP $\lambda
\times\rho$, R-NaD $\lambda\times K$). \emph{Both} methods have a stale-collapse
region (red): R-NaD at large $K$, GARIP at small $\rho$ (the two leftmost
columns). The point is location, not absence --- GARIP's red region sits at
averaging rates $\le\!0.005$, past its default $\rho\!=\!0.01$, whereas R-NaD's
reaches the commonly-used $K\!=\!200$. Matrix-grid and Leduc sweeps
(App.~\ref{app:sweeps}) confirm the same picture with an exact metric.}
\label{fig:sens-coin}
\end{center}

\FloatBarrier
\section{Honesty and limitations}

We are deliberately careful about what GARIP does \emph{not} do (extended version,
App.~\ref{app:limits}).
\begin{itemize}
  \item \textbf{No peak-performance claim; not collapse-immune.} GARIP only \emph{ties}
  R-NaD at each method's tuned best; the contribution is a \emph{wider safe basin}, not
  lower exploitability. It has its own staleness collapse --- a slow EMA
  ($\rho\!\approx\!10^{-3}$, $\lambda\!\ge\!0.5$) goes exploitable, symmetric to R-NaD's
  large-$K$ failure --- so the earlier ``no collapse region'' claim was wrong and we
  corrected it after a fair symmetric-$\rho$ sweep; the default $\rho\!=\!10^{-2}$ stays
  an order of magnitude clear.
  \item \textbf{Proxy metric and scope.} Exact metrics cover matrix games, Kuhn and
  Leduc; deep RL uses a best-response proxy that can reward \emph{passivity} on
  sparse-reward spatial games (so we use turn-based board games, which always terminate).
  The search-free learner does not reach robust play on $11\times11$ Hex, and Animal
  Shogi (non-cycling) needs no regularizer --- both honest boundaries, not counterexamples.
  \item \textbf{Which R-NaD.} We test \emph{fixed}-$(\lambda,K)$ R-NaD (standalone
  regularizer); DeepNash's outer schedule that anneals across convergence is untested,
  and we expect it trades staleness for annealing-sensitivity.
  \item \textbf{Local last iterate proven; global conjectured, with a failure mode.}
  Prop.~\ref{prop:localli} gives local convergence and Lemma~\ref{lem:klconv} reduces the
  global version to $\bar x_t\!\to$ Nash --- which our sweeps show \emph{fails} at large
  $\beta$ via premature consensus (Conjecture~\ref{conj:li}). This is a
  \emph{basin-of-attraction} failure, not a fixed-point/annealing one:
  Prop.~\ref{prop:fp}'s unbiasedness is untouched, and the deep-RL operating point has no
  large-$\beta$ analog. We claim no acceleration --- on an already-convergent base the
  anchor is a cost.
\end{itemize}

\section{Conclusion}

GARIP anchors self-play to the policy's running average --- a moving reference
that, like R-NaD's periodic snapshot, gives last-iterate Nash without annealing,
and matches R-NaD's peak performance. On robustness the two differ in their
reference's lag \emph{profile}: collapse empirically tracks the peak stale force,
and among causal averages the flat profile (the running average) is the
peak-minimizer (Prop.~\ref{prop:peakmin}), while a snapshot's peak is twice its mean.
This makes the running average the better \emph{default}: over the full grid the
aggregate collapse rates are indistinguishable, but at conventional parameterizations
--- R-NaD's mean-like $K$ under-counting its peak --- the standard $\rho\!=\!10^{-2}$
and $K\!=\!200$ collapse in $0/40$ vs $10/40$ seeds (including the $\lambda\!=\!2$
stress). The two tie at matched \emph{peak} lag, so a snapshot recovers safety by
shortening $K$ --- but GARIP gives the safe parameterization for free, where R-NaD
asks the user to know to do it. Its boundaries: an \emph{anticipatory}
(negative-weight) reference does better still on the stale side
(App.~\ref{sec:extrap}), and the edge over R-NaD appears only where naive self-play
cycles (Coin Game, Connect Four, Othello), not where it converges (Animal Shogi) or
where no method does (Hex). Within that scope, the running average is the reference
we would reach for, and the lag-profile lens says why.

\section*{Code Availability}
All code (the GARIP method and baselines, the tabular/poker/deep-RL experiments, the
local-stability and per-mode verification scripts, the premature-consensus sweep, and
the \LaTeX{} source) is available at \url{https://github.com/highcansavci/garip}.
Supplementary material (proofs, the discrete-dynamics derivation, the anticipatory-reference
study, and the $\beta^\star(d)$ phase boundary) is bundled in the appendices there.

\section*{Use of AI-Assisted Technologies}
The author used Claude (Anthropic), including Claude Code, to assist with the experimental
code, manuscript review and revision, and formatting. All AI-assisted content was verified
by the author, who is solely responsible for the methodology, results, and conclusions of
this paper.

\appendix

\section{Proofs}
\label{app:proofs}

\begin{proof}[Proof of Prop.~\ref{prop:peakmin} (peak-lag minimality)]
$\ell_{\max}=\max_t\ell_t\ge T^{-1}\!\sum_t\ell_t=\bar\ell$, with equality iff all
$\ell_t$ coincide. The EMA's weights $w_{t,s}=\rho(1-\rho)^{t-s}$ give the stationary
mean age $(1-\rho)/\rho$, constant in $t$; a snapshot held for $K$ steps has
$\ell_t\in\{0,\dots,K\!-\!1\}$, so $\bar\ell=\tfrac{K-1}{2}$ and
$\ell_{\max}=K\!-\!1=2\bar\ell$.
\end{proof}

\begin{proof}[Proof of Lemma~\ref{lem:contract} (anchor contraction)]
$x_{t+1}-\bar x_t = (1-\beta)x_{t+\frac12}+\beta\bar x_t-\bar x_t
= (1-\beta)(x_{t+\frac12}-\bar x_t)$.
\end{proof}

\begin{proof}[Proof of Prop.~\ref{prop:fp} (fixed points are Nash)]
At a fixed point $x_{t+1}=x_t=:x$ with $\bar x_t=x$, Lemma~\ref{lem:contract} forces
$x_{t+\frac12}=x$: the OMWU half-step fixes $x$. With $g_t=g_{t-1}=Ay$ this is
$x\propto x\odot\exp(\eta\,Ay)$, so $(Ay)_i$ is constant on $\operatorname{supp}(x)$.
For interior $x$ this holds for all $i$, i.e.\ $y$ makes the row player indifferent;
symmetrically $x$ makes the column player indifferent, so $(x,y)$ is Nash, giving
(ii); a Nash satisfies these on its support with the anchor inactive, giving (i).
For (iii), the magnet contribution at the fixed point is
$\beta(\bar x_t-x_{t+\frac12})=0$.
\end{proof}

\begin{proof}[Proof of Prop.~\ref{prop:localli} (local last-iterate convergence)]
The anchor $x_{t+1}=(1-\beta)x_{t+\frac12}+\beta\bar x_t$ and average
$\bar x_{t+1}=(1-\rho)\bar x_t+\rho x_{t+1}$ are coordinate-wise identical scalars, so
they commute with the base operator and the linearization splits per base mode $M$,
with $x_{t+\frac12}-x^\star=M(x_t-x^\star)$; the $(1-\beta)$ factor on the policy row is
Lemma~\ref{lem:contract}, and $\det J=(1-\beta)M[(1-\rho)+\rho\beta]-\beta\rho(1-\beta)M
=(1-\rho)(1-\beta)M$. \emph{Fast mode:} at $\rho=0$ the eigenvalues are exactly
$(1-\beta)M$ (modulus $(1-\beta)\mu$) and the marginal averaging mode $1$; the former
carries the mechanism --- a recurrent base ($\mu=1$) is strictly contracted to
$(1-\beta)<1$. \emph{Slow mode:} expanding the marginal root $\lambda=1+\rho\,c$ in the
characteristic polynomial $\lambda^2-(\mathrm{tr}\,J)\lambda+\det J$ gives, to first
order,
\[
\lambda_{\mathrm{slow}}=1-\rho\,\frac{(1-\beta)(1-M)}{1-(1-\beta)M}+O(\rho^2),
\]
so $|\lambda_{\mathrm{slow}}|=1-\rho\,\mathrm{Re}\,K+O(\rho^2)$ with
$K=(1-\beta)(1-M)/(1-(1-\beta)M)$. Writing $M=1-\delta+i\omega$ ($\delta,\omega=\Theta(\eta)$,
$\delta>0$), $K\to(1-\beta)\delta/\beta>0$ as $\eta\to0$, so $\mathrm{Re}\,K>0$ and
$|\lambda_{\mathrm{slow}}|<1$ for small $\eta$. (The naive real form
$1-\rho(1-\mu)(1-\beta)$ is \emph{not} correct --- $M$ is a rotation, so the coefficient
is complex.) For the Ces\`aro average of Sec.~\ref{sec:analysis} (weight
$1/(t{+}2)\!\to\!0$), $u_t:=x_t-\bar x_t$ obeys
$u_{t+1}=(1-\beta)M\,u_t+(1-\beta)(M-I)\bar\delta_t$ with
$\bar\delta_t:=\bar x_t-x^\star$ slowly varying; $u_t\to0$ geometrically, and the only
frozen anchor consistent with $u\!=\!0$ solves $(M-I)\bar\delta=0$, i.e.\
$\bar\delta=0$ since $1\notin\mathrm{spec}(M)$ --- so $\bar x_t\to x^\star$, closing
step~(a) locally. For the one-step base (MWU --- the recurrent, deep-RL regime) the
anchor and average are scalars on the \emph{entire} mode space, so the
block-diagonalization is \emph{exact} for any $m\times n$ game and
$\rho(\text{full})=\max_k\rho(J(M_k))$ holds identically (confirmed to
finite-difference precision, $<\!10^{-10}$, across sizes up to $8\times8$ with $10$
random payoffs each). GARIP's optimistic half-step carries a lagged gradient the anchor
does not touch, enlarging each per-mode block from $2\times2$ to $3\times3$; the
structure --- a fast mode scaled by $1-\beta$, a slow average mode near $1$ --- is
unchanged, and the $2\times2$ form agrees with the full spectrum to $<\!4\times10^{-4}$.
\end{proof}

\begin{proof}[Proof of Lemma~\ref{lem:klconv} (potential-level contraction)]
$x_{t+1}=(1-\beta)x_{t+\frac12}+\beta\bar x_t$ with convexity gives the first bound;
$\bar x_t=\frac1t\sum_{s<t}x_s$ with the same convexity (Jensen) gives the second.
\end{proof}

\section{A discrete-dynamics mechanism for the peak-force collapse law}
\label{app:discrete}
This appendix expands the footnote-level claim that the collapse law has a derivable
(discrete, not continuous) origin.
The collapse law is not pure phenomenology; it has a derivable origin, just not a
\emph{continuous} one. Linearizing a $2$-state game in continuous time, the stale
magnet is a delayed term $\dot z=\mu_0 z+\lambda(z(t\!-\!\tau)-z(t))$ on the
converging spiral $\mu_0=-\delta\pm i\omega$; its characteristic equation
$\mu=\mu_0+\lambda(e^{-\mu\tau}-1)$ has \emph{no} imaginary-axis root
($\mathrm{Re}$ at $\mu=i\nu$ is $-\delta+\lambda(\cos\nu\tau-1)\le-\delta<0$), so in
the continuous limit the magnet \emph{damps} --- Nash never loses stability. The
\emph{discrete} map, which is what runs, does collapse. With one mode
$z_{t+1}=(1-\beta)Mz_t+\beta\,r_t$, $M=(1-\delta)e^{i\omega}$ ($|M|<1$), the past
iterate $z_{t-\ell}=M^{-\ell}z_t$ had \emph{larger} amplitude $|M|^{-\ell}\!>\!1$. A
snapshot reaches the single peak-lag state $M^{-K}$ (full amplification
$|M|^{-K}\!\approx\!e^{\delta K}$), so a strong/slow magnet drives the step gain
$g=(1-\beta)M+\beta M^{-K}$ past $|g|\!=\!1$ --- a discrete \emph{overshoot},
governed by the \emph{peak} reach; the EMA smears its reach over lags
($\sum_s\rho(1-\rho)^sM^{-s}$), so its gain is bounded by the flat profile, not a
deep peak. This reproduces the peak-force law and its peak-vs-mean asymmetry. The toy
in fact \emph{over}-predicts --- an $e^{\delta K}$ gap, exponential in lag, sharper
than the roughly-linear empirical curves of Fig.~\ref{fig:lag}; the gap between the
two is exactly the phase cancellation and finite-step saturation the linear one-mode
model omits, so the over-prediction is a consistency check, not a contradiction. Only
the \emph{exact} threshold is therefore phase-dependent, which is why we calibrate the
law empirically (Fig.~\ref{fig:lag}) and prove the reference-shape optimality on top
of it (Prop.~\ref{prop:peakmin}).

\section{Beyond causal averages: an anticipatory reference}
\label{sec:extrap}
Proposition~\ref{prop:peakmin} is optimal only \emph{within} causal averages, and
Remark~\ref{rem:causal} flagged that an \emph{extrapolating} reference (negative
weights, lower effective lag) might beat it. We tested it directly. A double-EMA
magnet $(1{+}\gamma)\,\bar\theta-\gamma\,\bar{\bar\theta}$ ($\bar{\bar\theta}$ the EMA
of $\bar\theta$, gain $\gamma{=}1$) \emph{leads} the policy --- a negative-weight
reference, still causal but outside the \emph{convex} class --- which we swept
against plain GARIP on the
deep-RL $(\lambda,\rho)$ collapse grid (Table~\ref{tab:extrap}; $18$ runs/cell over
$\lambda\!\in\!\{0.5,1,2\}$).

\begin{table}[h]
\centering
\caption{Collapse rate (exploit return $>0$) vs.\ averaging rate $\rho$ for the
causal running average (GARIP) and the anticipatory double-EMA magnet
(Coin Game; lower $=$ safer, bold). Anticipation widens the safe basin on the
\emph{stale} (slow-$\rho$) side but overshoots at \emph{fast} $\rho$. Here each cell
is $18$ runs ($\lambda\!\in\!\{0.5,1,2\}\times6$ seeds), \emph{not} the $40$ of
Table~\ref{tab:sens-coin} (which also includes the safe $\lambda\!=\!0.25$), so
GARIP's rates run higher; Wilson $95\%$ CIs are wide ($\sim\!\pm0.2$), e.g.\ the
$\rho\!=\!0.05$ overshoot is $2/18$, CI $[0.03,0.33]$.}
\label{tab:extrap}
\small
\resizebox{\columnwidth}{!}{%
\begin{tabular}{lcc}
\toprule
$\rho$ & GARIP (causal avg.) & extrap.\ (anticipatory) \\
\midrule
$0.00125$ & $1.00$ & $\mathbf{0.56}$ \\
$0.0025$  & $0.67$ & $\mathbf{0.00}$ \\
$0.005$   & $0.33$ & $\mathbf{0.00}$ \\
$0.01$ (default) & $0.00$ & $0.00$ \\
$0.02$    & $0.00$ & $0.00$ \\
$0.05$    & $\mathbf{0.00}$ & $0.11$ \\
\bottomrule
\end{tabular}}
\end{table}

The prediction holds, with a twist. Anticipation \emph{subtracts} lag, so it widens
the safe basin exactly on the stale side: collapse falls $67\%\!\to\!0\%$ at
$\rho{=}0.0025$ and $33\%\!\to\!0\%$ at $\rho{=}0.005$, pushing the boundary from
$\rho\!\approx\!0.005$ down to $\rho\!\approx\!0.0025$. But over-anticipating a
\emph{fast} average overshoots --- at $\rho{=}0.05$ the extrapolated magnet newly
collapses ($0\%\!\to\!11\%$). The two are therefore \emph{complementary}, not
dominating: plain GARIP is safe for $\rho\!\in\![0.01,0.05]$, the anticipatory magnet
for $\rho\!\in\![0.0025,0.02]$, and GARIP's default $\rho{=}0.01$ sits safely in
both. So Prop.~\ref{prop:peakmin} is genuinely breakable by leaving the convex class
(negative weights) --- exactly as scoped --- and the peak-lag picture predicts both halves
(anticipation helps when lag is large, hurts when it is already small): the running
average is the robust default, anticipation is the right tool when one is forced to
slow averaging, and an adaptive magnet that tunes $\gamma$ to the operating lag is a
natural next step.

\section{The premature-consensus phase boundary}
\label{app:phase}
At a constant anchor the reduction of Conjecture~\ref{conj:li} (global last iterate
$\Leftarrow\bar x_t\!\to$ Nash) can fail: a strong anchor drives $x_t$ and $\bar x_t$
into agreement \emph{before} the average reaches Nash, freezing the iterate at a
non-Nash consensus. We map the boundary on antisymmetric $d\times d$ games (interior
Nash), classifying a run as stalled when exploitability is frozen high while
$\Vert x_t-\bar x_t\Vert\!\to\!0$. The critical anchor $\beta^\star(d)$ above which a
game stalls falls with size:
\begin{center}
\begin{tabular}{lccccc}
\toprule
$d$ & $\le 4$ & $6$ & $8$ & $10$ & $12$ \\
\midrule
$\beta^\star(d)$ & none ($\le0.4$) & $0.2$ & $0.1$ & $0.05$ & $0.02$ \\
\bottomrule
\end{tabular}
\end{center}
So $\beta^\star$ shrinks as the game grows and the default $\beta\!=\!0.02$ sits below
the boundary through $d\!\approx\!10$. Pure OMWU ($\beta\!=\!0$) reaches Nash to
$\sim\!10^{-7}$ on these games: on an already-convergent base a strong anchor
over-regularizes --- the tabular face of the Animal-Shogi lesson (a magnet is
unnecessary on a non-cycling game). The deep-RL operating point has no analog: there
collapse is reference \emph{staleness} (slow $\rho$ / long $K$), not anchor
\emph{strength}, and the anchor is load-bearing rather than a liability.

\section{Matrix curves, poker details, and the RPS trajectory}
\label{app:poker}

\begin{figure*}[t]
\centering
\includegraphics[width=0.7\linewidth]{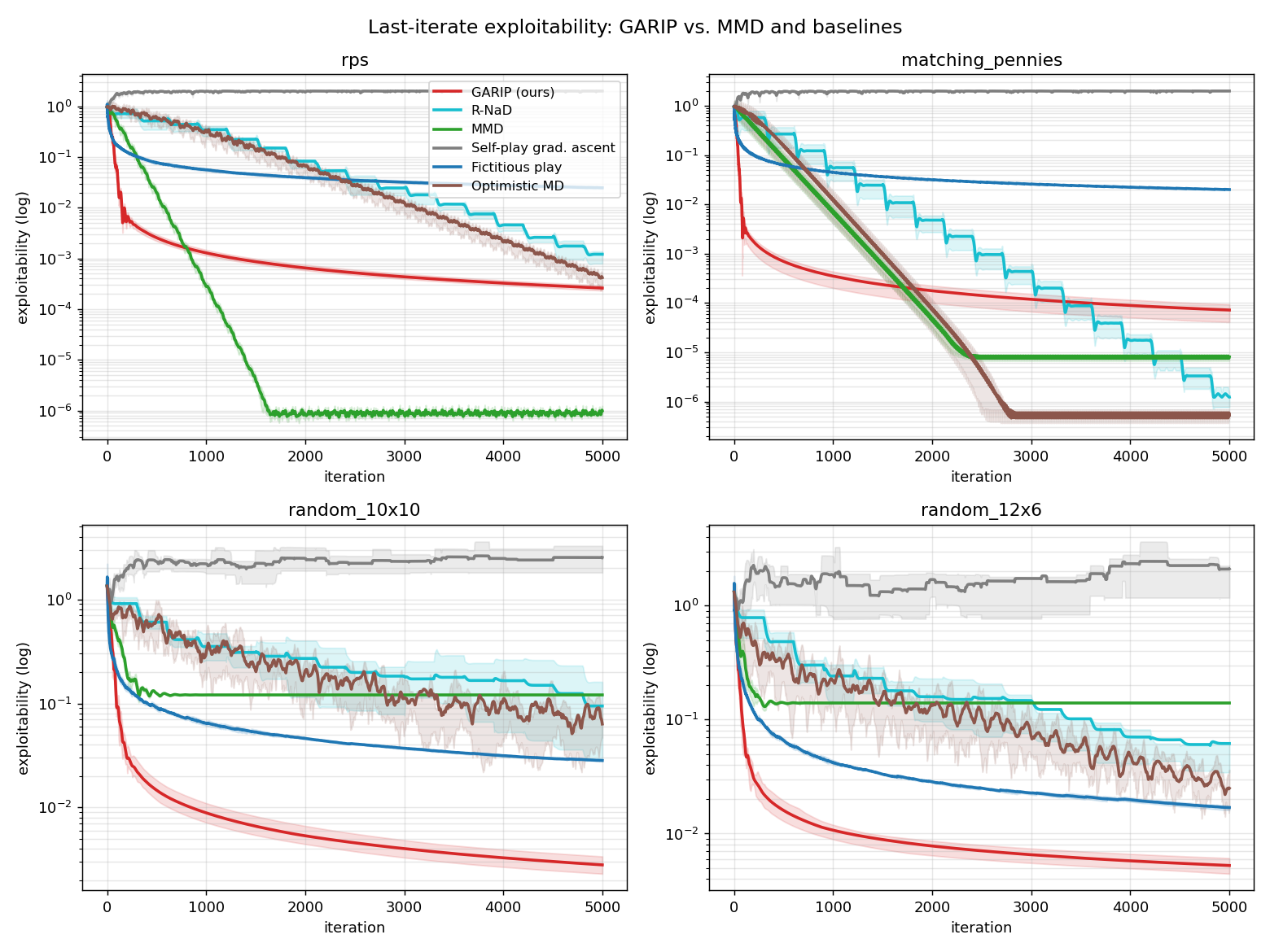}
\caption{Last-iterate exploitability vs.\ iteration on four matrix games
(10 seeds, shaded IQR). GARIP (red) reaches the lowest exploitability on the
random games; MMD plateaus at the QRE; self-play gradient ascent stays at
$\sim$2.}
\label{fig:curves}
\end{figure*}

\begin{figure}[h]
\centering
\includegraphics[width=0.7\linewidth]{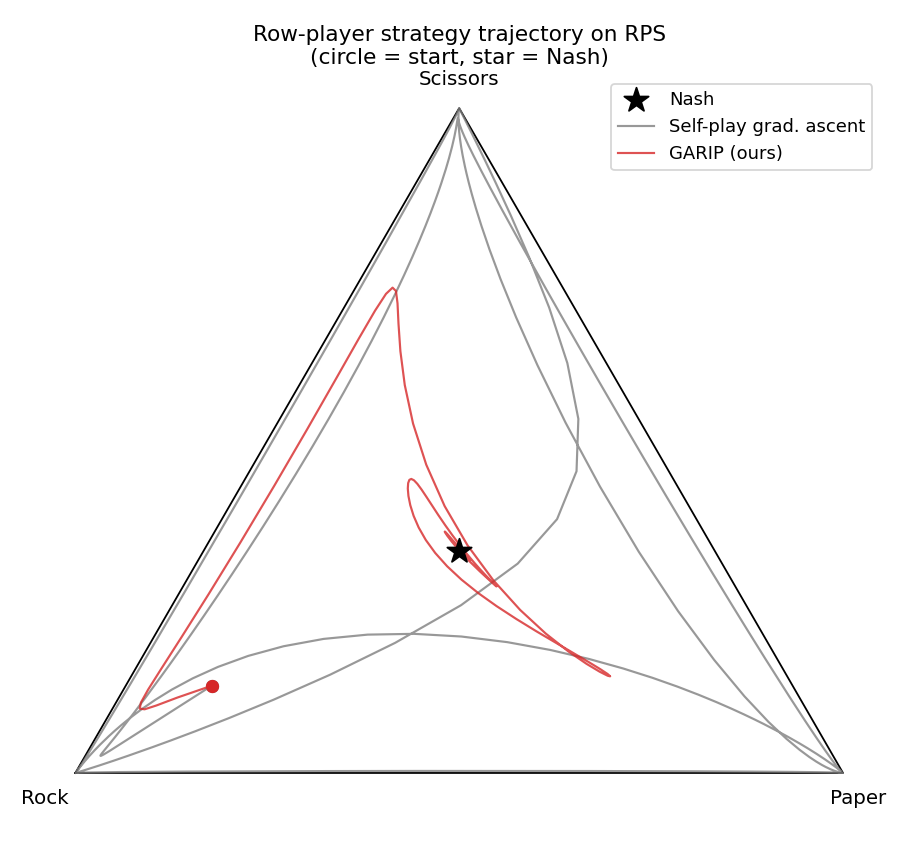}
\caption{Row-player strategy trajectory on the RPS simplex. Naive self-play
gradient ascent (grey) orbits the boundary; GARIP (red) spirals into Nash
(star).}
\label{fig:simplex}
\end{figure}

\paragraph{Kuhn (neural).} Each player is a small MLP over information-set
features, trained through the differentiable game tree; exploitability is exact
(brute force over the 64 pure strategies). Using the unified magnet framework
(Table~\ref{tab:kuhn}), GARIP (0.076) is the best regularized method, ahead of
R-NaD (0.30) and MMD (0.54), and approaches exact-best-response fictitious play
(0.043). These are each method's \emph{default} config --- we did not separately
tune R-NaD's $(\alpha,K)$ on the poker games, so the gap is default-vs-default and a
tuned R-NaD would likely close much of it, exactly as it does on the matrix games;
we do not claim a tuned-best poker win.

\begin{table}[h]
\centering
\caption{Neural Kuhn poker: average-strategy exact exploitability
(10 seeds). Same ordering throughout: moving $>$ snapshot $>$ fixed magnet.}
\label{tab:kuhn}
\small
\begin{tabular}{lc}
\toprule
method & exploitability \\
\midrule
fictitious play (exact BR)        & 0.043 \\
\textbf{GARIP (moving magnet)}    & \textbf{0.076} \\
R-NaD (periodic snapshot)         & 0.301 \\
naive self-play                   & 0.373 \\
MMD (fixed magnet)                & 0.538 \\
\bottomrule
\end{tabular}
\end{table}

\paragraph{Leduc (tabular, exact).} Leduc has 288 information sets and an exact
recursive best response and CFR solver (game value lands on the known
$-0.0856$). Because raw-gradient \emph{neural} self-play does not converge for
the magnet baselines here --- the reach/credit-assignment failure the
counterfactual target was built to fix --- we run the methods tabularly with
mirror ascent on exact counterfactual values. GARIP is fastest and lowest
\emph{at this iteration budget} (Table~\ref{tab:leduc}, Figure~\ref{fig:leduc}).
It nominally sits below CFR at 3000 iterations, but we do \emph{not} read this as
beating CFR: CFR's averaged strategy converges slowly early ($O(1/\sqrt{T})$),
and a mirror-ascent step is not the same unit of compute as a CFR iteration, so
CFR is a \emph{reference point}, not a baseline GARIP outperforms --- with more
sweeps it keeps falling to $\sim$0.04. As on Kuhn, these are default configs and
R-NaD was not separately tuned here.

\begin{table}[h]
\centering
\caption{Leduc hold'em: average-strategy exact exploitability (8 seeds,
3000 iterations). Default configs; CFR is shown at a matched iteration count as a
reference point (not converged --- it falls to $\sim$0.04 with more sweeps), not
as a baseline GARIP outperforms.}
\label{tab:leduc}
\small
\begin{tabular}{lc}
\toprule
method & exploitability \\
\midrule
\textbf{GARIP (moving magnet)} & \textbf{0.074} \\
CFR (reference)                & 0.116 \\
R-NaD (periodic snapshot)      & 0.227 \\
MMD (fixed magnet)             & 0.552 \\
naive self-play                & 0.741 \\
\bottomrule
\end{tabular}
\end{table}

\begin{figure}[h]
\centering
\includegraphics[width=0.7\linewidth]{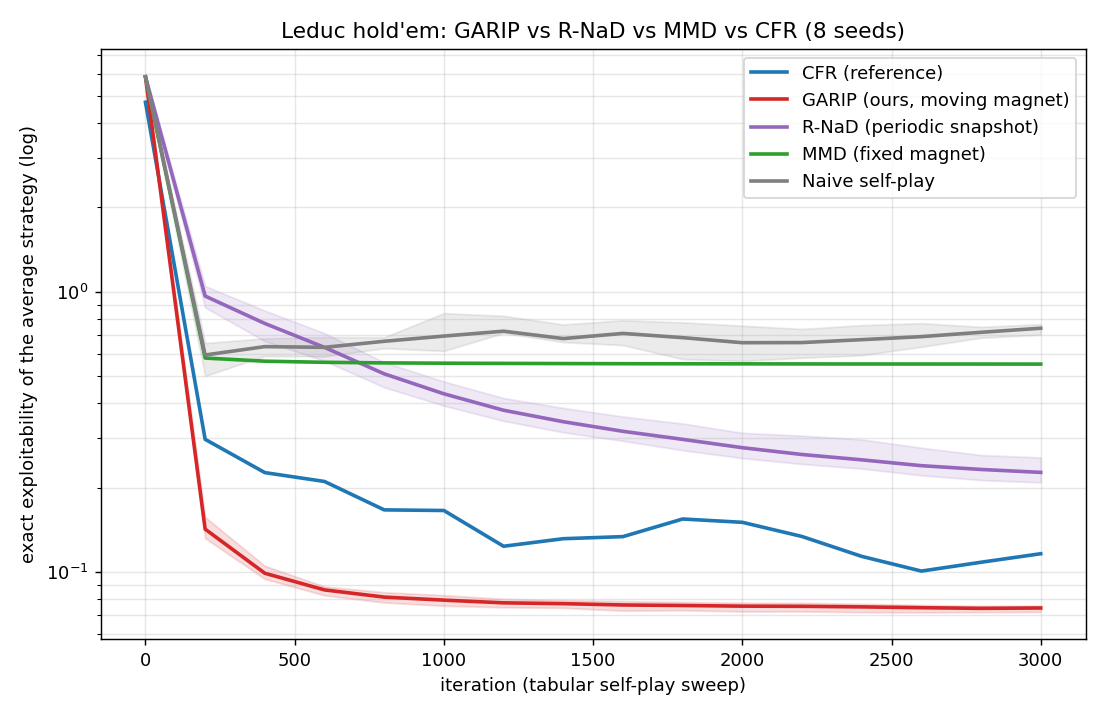}
\caption{Leduc hold'em: exact exploitability of the average strategy
(8 seeds). GARIP converges fastest; MMD's fixed magnet plateaus at the QRE;
naive self-play cycles.}
\label{fig:leduc}
\end{figure}

\section{Additional robustness sweeps}
\label{app:sweeps}
The matrix-grid and Leduc confirmations of the Coin Game sensitivity picture
(\S\ref{sec:robust}), deferred from the main text.

\begin{center}
\small
\captionof{table}{Robustness over the $5\times5$ matrix grid (GARIP $(\eta,\beta)$,
R-NaD $(\alpha,K)$); lower median / IQR / worst-case exploitability and a higher
converged fraction indicate a wider safe basin.}
\label{tab:sens-matrix}
\resizebox{\columnwidth}{!}{%
\begin{tabular}{lcc}
\toprule
over the grid & \textbf{GARIP} $(\eta,\beta)$ & R-NaD $(\alpha,K)$ \\
\midrule
median exploitability   & \textbf{0.031} & 0.332 \\
spread (IQR)            & \textbf{0.126} & 0.450 \\
worst config           & \textbf{0.772} & 1.209 \\
configs converged ($<\!0.05$) & \textbf{56\%} & 20\% \\
\bottomrule
\end{tabular}}

\medskip
\captionof{table}{Robustness over the Leduc grid, exact exploitability (6 seeds).
R-NaD never reaches near-Nash; GARIP's only failure is generic large-step
instability.}
\label{tab:leduc-collapse}
\resizebox{\columnwidth}{!}{%
\begin{tabular}{lcc}
\toprule
Leduc & \textbf{GARIP} $(\eta,\beta)$ & R-NaD $(\alpha,K)$ \\
\midrule
median                  & \textbf{0.080} & 0.427 \\
configs near Nash ($<\!0.1$) & \textbf{64\%} & 0\% \\
worst config            & \textbf{2.21}  & 2.70 \\
\bottomrule
\end{tabular}}

\medskip
\includegraphics[width=0.7\linewidth]{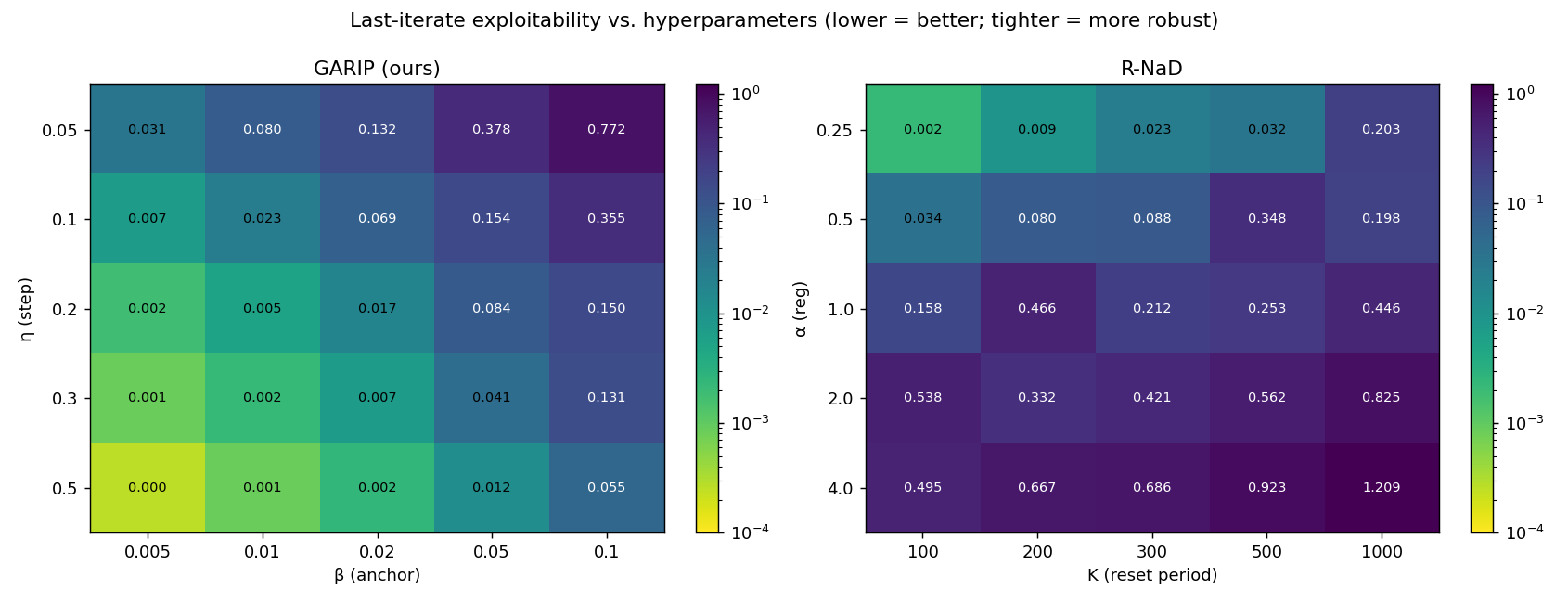}
\captionof{figure}{Matrix-game hyperparameter sensitivity. GARIP $(\eta,\beta)$
has a wide low-exploitability basin; R-NaD $(\alpha,K)$ degrades sharply outside
a narrow band.}
\label{fig:sens}

\medskip
\includegraphics[width=0.7\linewidth]{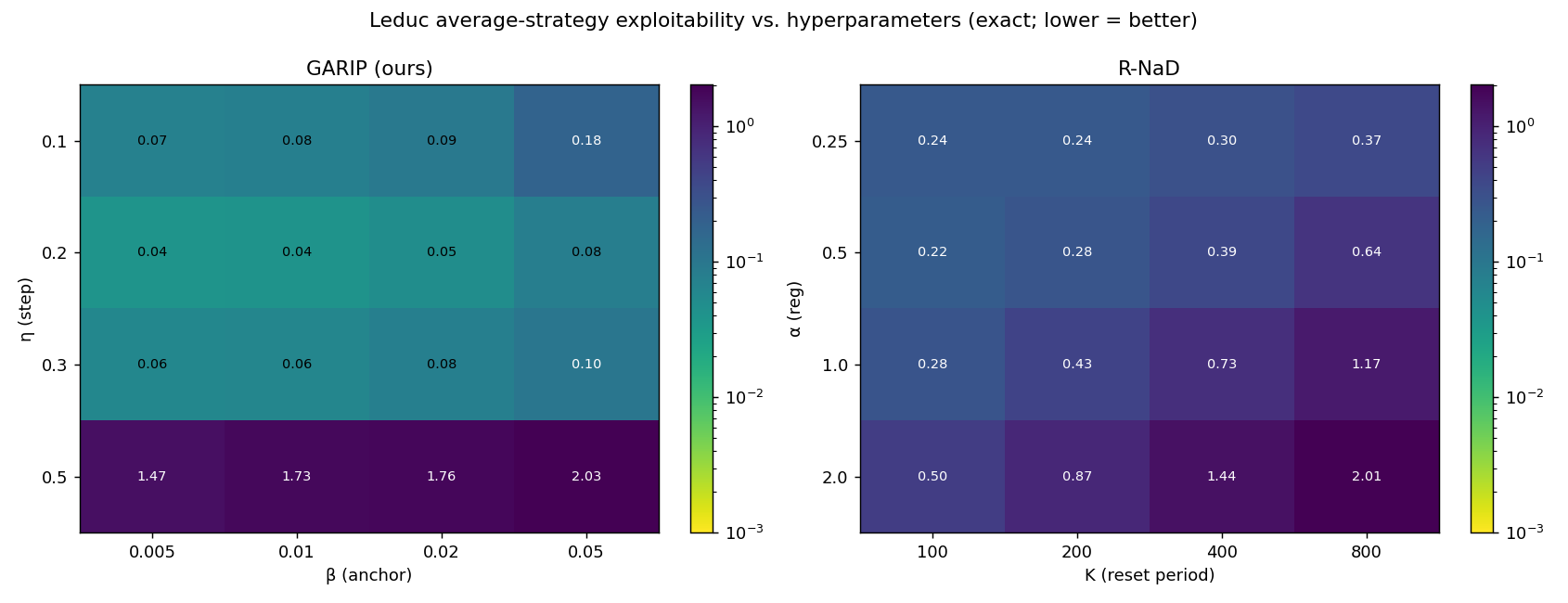}
\captionof{figure}{Leduc collapse heatmaps (exact exploitability, 6 seeds). R-NaD
shows the predicted monotonic $\alpha K$ gradient and never reaches near-Nash;
GARIP's only failure is generic large-step-size instability ($\eta\!=\!0.5$).}
\label{fig:leduc-collapse}
\end{center}

\paragraph{Alternative references we tested.} Beyond the causal running average and
the anticipatory double-EMA (App.~\ref{sec:extrap}), we tried two further mechanisms
and report them for completeness: a \emph{reciprocal-cycle} reference (anchor to the
double quantal-best-response $F(G(x))$, the literal CycleGAN cycle term) and a
\emph{drift-adaptive} averaging rate. Neither improved on the running average ---
the reciprocal cycle inherits the game's rotation and destabilizes, and adaptive
anticipation either leaves the simplex or overshoots --- consistent with
Prop.~\ref{prop:peakmin}: among causal references the running average is already
peak-lag optimal, and leaving that class buys only the narrow stale-side gain of
App.~\ref{sec:extrap}.

\section{Extended honesty and limitations}
\label{app:limits}
The full version of the limitations summarized in the main text.
\begin{itemize}
  \item \textbf{No peak-performance claim.} GARIP only \emph{ties} R-NaD at each
  method's tuned best (matrix games and Coin Game). The contribution is a
  \emph{wider safe basin} (collapse boundary past the default rate), not lower
  exploitability and not collapse-immunity.
  \item \textbf{GARIP has its own staleness collapse.} A slow EMA
  ($\rho\!\approx\!10^{-3}$) at $\lambda\!\ge\!0.5$ goes exploitable, symmetric to
  R-NaD's large-$K$ failure (Sec.~\ref{sec:theory}); GARIP is \emph{not}
  collapse-free. Its default $\rho\!=\!10^{-2}$ is robust, and the collapse region
  lies an order of magnitude away --- but the ``no collapse region'' claim of an
  earlier draft was wrong, and we corrected it after a fair (symmetric-$\rho$)
  sweep. Separately, the largest matrix step size $\eta\!=\!0.5$ blows up on
  Leduc --- generic step-size instability, a single isolated row.
  \item \textbf{Scope and the proxy metric.} Exact metrics cover matrix games,
  Kuhn and Leduc; the deep-RL evidence uses a fixed-budget best-response proxy.
  This proxy is delicate: on the simultaneous-move Coin Game it needs a strong
  enough best responder, and on a sparse-reward spatial game we found it can
  reward \emph{passivity} (a non-engaging policy is unexploitable by
  construction), so we did not use such games. Turn-based board games are the fix
  --- they always terminate with a result, so the metric cannot be gamed by
  inaction --- which is why the board panel (Sec.~\ref{sec:boardgames}) is the
  primary deep evidence. Its limits are honest: the search-free PPO learner does
  not reach robust play on $11\times11$ Hex even at $10{,}000$ updates (the metric
  saturates), and on a small non-cycling game (Animal Shogi) regularization is
  simply unnecessary. Stronger learners (search/AlphaZero-style) and
  Stratego-scale games remain future work.
  \item \textbf{Which R-NaD is under test.} We test \emph{fixed}-$(\lambda,K)$
  R-NaD: a periodic snapshot reset every $K$ steps under a constant KL weight
  $\lambda$ --- the form used as a standalone self-play regularizer, and the clean
  point of comparison for ``which moving reference.'' The full R-NaD behind
  DeepNash additionally runs an \emph{outer schedule} that updates the
  regularization policy and anneals the reward transform across convergence,
  partly to manage exactly the staleness we exploit. Our critique is of the
  fixed-schedule instantiation; whether it survives the scheduled version is open.
  We expect the schedule trades staleness for annealing-sensitivity --- a fixed
  $K$ becomes a schedule to tune, which is itself a hyperparameter-robustness cost
  --- but we have not measured this, and say so plainly.
  \item \textbf{Local last iterate proven; global is conjectured.}
  Sec.~\ref{sec:analysis} proves the structural facts --- the anchor contracts
  deviation from the average by exactly $1-\beta$ (Lemma~\ref{lem:contract}), and is
  unbiased so that fixed points are Nash, unlike MMD's QRE-biased fixed magnet
  (Prop.~\ref{prop:fp}) --- and now a \emph{local} last-iterate theorem
  (Prop.~\ref{prop:localli}): the anchor scales the base's rotation by $1-\beta$,
  crossing the stability boundary and turning a recurrent base into a contraction, so
  Nash is locally asymptotically stable in the last iterate at $\beta\!>\!0$. The
  potential-level reduction (Lemma~\ref{lem:klconv}) collapses the \emph{global} version
  to a single estimate, $\bar x_t\to$ Nash --- but our long-horizon sweeps show this can
  \emph{fail} (App.~\ref{app:phase}): at larger $\beta$ on larger interior-Nash games the
  anchor drives a \emph{premature consensus} ($x_t\!\approx\!\bar x_t$ at a non-Nash
  point) and the last iterate stalls, while pure OMWU ($\beta\!=\!0$) reaches Nash on the
  same games. So we do not claim global convergence or acceleration: the anchor is a
  robustness device for \emph{cycling} bases, not a last-iterate accelerator, and on an
  already-convergent base it is a cost. This is a \emph{basin-of-attraction} failure at
  large $\beta$, not a fixed-point or annealing issue --- Prop.~\ref{prop:fp}'s
  unbiasedness (no annealing needed) is untouched, and the deep-RL operating point has no
  large-$\beta$ analog (there collapse is reference staleness, not anchor strength).
  \item \textbf{Earlier mis-report, corrected.} An initial 2-seed, narrow-grid
  run suggested the opposite robustness conclusion; it under-sampled R-NaD's
  stale-reference failure region. The 10-seed wider-grid results here supersede
  it.
\end{itemize}

\end{document}